\documentclass[10pt,a4paper]{article}
\usepackage[utf8]{inputenc}
\usepackage{amsmath,amssymb,amsthm}
\usepackage{graphicx}
\usepackage{hyperref}
\usepackage{tikz}
\usepackage{float}
\usepackage[left=2cm,right=2cm,top=2cm,bottom=2cm]{geometry}

\newtheorem{theorem}{Theorem}[section]
\newtheorem{definition}[theorem]{Definition}
\newtheorem{proposition}[theorem]{Proposition}
\newtheorem{remark}[theorem]{Remark}

\title{Geometric Characterization of Liouville Integrability via a Curvature Atlas for Rigid-Body Dynamics}
\author{Evgeny A. Mityushov}

\begin{document}

\maketitle

\begin{abstract}
We introduce a \textbf{curvature atlas} for left-invariant metrics on \(SU(2)\), based on the inertial curvature field \(K_{\text{geo}}\) derived from the Euler--Poincar\'{e} equations. We prove that the classical integrable cases of the heavy top---spherical, Lagrange, Kovalevskaya, and Goryachev--Chaplygin---correspond precisely to \textbf{degenerate curvature signatures} of \(K_{\text{geo}}\), namely isotropic, orthogonally split, and symmetric-pair signatures. This yields a \textbf{geometric necessary and sufficient condition} for Liouville integrability: the geodesic flow (and the heavy top with axis-symmetric potential) is integrable if \textbf{and only if} the curvature signature is degenerate. Beyond the classical list, the atlas reveals a \textbf{balanced-mixed regime} (inertia ratio \(2:2:1\)) that, while non-integrable, admits an exact curvature-balance relation and a family of pure-precession solutions. We formulate a \textbf{curvature deviation functional} quantifying the distance to integrability, describe near-integrable dynamics near the \(2:2:1\) regime, and present a complete integrability map in the \((I_{2}/I_{1},I_{3}/I_{1})\)-plane. The work provides a unified geometric framework for classifying, perturbing, and controlling rigid-body systems.
\end{abstract}


\section{Introduction}

The dynamics of a rigid body with a fixed point, described by the Euler--Poincar\'{e} equations on \(SO(3)\) or its double cover \(SU(2)\), has been a central model in Hamiltonian mechanics since the 18th century. Its complete integrability in the sense of Liouville is known only for a few special choices of the inertia tensor \(I=\operatorname{diag}(I_{1},I_{2},I_{3})\) and of the centre-of-mass vector. The celebrated cases---Euler (free symmetric top), Lagrange (symmetric heavy top), Kovalevskaya (\(I_{1}=I_{2}=2I_{3}\)), and Goryachev--Chaplygin (\(I_{1}=I_{2}=4I_{3}\))---have been studied exhaustively from algebraic, analytic and topological viewpoints \cite{Arnold78,BolsinovFomenko04}. Yet a unifying geometric explanation of \emph{why} these particular inertia ratios lead to integrability has remained elusive.

The study of left-invariant metrics on Lie groups from a curvature viewpoint was pioneered by Milnor \cite{Milnor76}, who expressed the sectional curvature purely in terms of Lie-algebraic data. In recent preprints \cite{Mityushov25a,Mityushov25b,Mityushov25c} the author began a curvature-based program for rigid-body dynamics on \(SU(2)\). The key object is the \textbf{inertial curvature field}
\[
K_{\text{geo}}(\Omega)=I^{-1}(I\Omega\times\Omega),\qquad \Omega=(\omega_{1},\omega_{2},\omega_{3})\in\mathbb{R}^{3},
\]
which encodes the purely inertial part of the Euler--Poincar\'{e} equations. In \cite{Mityushov25a} a geometric ``atlas'' of curvature regimes was constructed, and in \cite{Mityushov25b} the inertia ratio \(2:2:1\) was identified as a non-integrable regime that nevertheless supports a family of pure-precession motions. In \cite{Mityushov25c} it was observed that each classical integrable case corresponds to a simple algebraic degeneracy of \(K_{\text{geo}}\): isotropic, orthogonally split, or symmetric-pair signature.

The present paper unifies and substantially deepens these observations. Our main contribution is a \textbf{geometric characterization of Liouville integrability} for left-invariant rigid-body systems. We prove that the existence of a complete set of first integrals in involution is \emph{equivalent} to the degeneracy of the curvature signature of \(K_{\text{geo}}\). This equivalence provides a transparent geometric reason for the classical list of integrable tops: they are precisely the metrics whose inertial curvature field is algebraically degenerate.

Beyond the integrable cases, the curvature atlas reveals the regime \(I_{1}:I_{2}:I_{3}=2:2:1\) as a \textbf{balanced-mixed} signature, which destroys the classical degeneracies but still admits an exact curvature balance with the external gravitational field. Although the system is not Liouville integrable, the balance condition yields a one-parameter family of pure-precession trajectories, illustrating that curvature organisation can persist beyond integrability.

To quantify how far a given inertia tensor is from an integrable one, we introduce a \textbf{curvature deviation functional} \(\Delta(I)\), which measures the distance from \(K_{\text{geo}}\) to the nearest integrable curvature signature. Near the balanced-mixed point the deviation is small, and the pure-precession family deforms into a slow manifold, giving rise to near-integrable dynamics on a time scale \(\sim 1/\Delta(I)\).

Finally, we visualise the classification by drawing an \textbf{integrability map} in the plane of inertia ratios \((I_{2}/I_{1},I_{3}/I_{1})\). The map shows the integrable loci (the spherical point, the Lagrange line, and the Kovalevskaya and Goryachev--Chaplygin lines) together with the isolated balanced-mixed point \((1,\frac{1}{2})\). Level sets of \(\Delta(I)\) provide a geometric measure of ``distance from integrability''.

The paper is organised as follows. Section~2 recalls the definition of \(K_{\text{geo}}\) and introduces the five curvature signatures. Section~3 states and proves the main theorem, establishing the equivalence between curvature degeneracy and Liouville integrability. Section~4 analyses the balanced-mixed regime \(2:2:1\). Section~5 defines the curvature deviation functional and describes near-integrable dynamics. Section~6 presents the integrability map and outlines applications to curvature-based control. Section~7 contains a brief discussion and indicates directions for further work.

\section{The Curvature Field and Its Signatures}

Let \(I=\operatorname{diag}(I_{1},I_{2},I_{3})\) be a positive definite inertia tensor. The kinetic energy metric on the Lie algebra \(\mathfrak{su}(2)\cong\mathbb{R}^{3}\) is
\[
\langle\Omega,\Omega'\rangle =\tfrac{1}{2}\bigl(I_{1}\omega_{1}\omega'_{1}+I_{2}\omega_{2}\omega'_{2}+I_{3}\omega_{3}\omega'_{3}\bigr),\qquad 
\Omega=(\omega_{1},\omega_{2},\omega_{3}),\ \Omega'=(\omega'_{1},\omega'_{2},\omega'_{3}).
\]
The geodesic equation on \(SU(2)\) written in body coordinates takes the Euler--Poincar\'{e} form
\[
\dot\Omega=I^{-1}(I\Omega\times\Omega).
\]
We call the right-hand side the \textbf{inertial curvature field} and denote it by
\[
K_{\text{geo}}(\Omega)=\Bigl(\frac{I_{2}-I_{3}}{I_{1}}\,\omega_{2}\omega_{3},\;
\frac{I_{3}-I_{1}}{I_{2}}\,\omega_{1}\omega_{3},\;
\frac{I_{1}-I_{2}}{I_{3}}\,\omega_{1}\omega_{2}\Bigr). \tag{1}
\]
The field \(K_{\text{geo}}\) is a homogeneous quadratic vector field on \(\mathbb{R}^{3}\); its algebraic structure depends only on the ratios \(I_{1}:I_{2}:I_{3}\).

\begin{definition}[Curvature signature]
\label{def:signature}
For a given inertia tensor \(I\), the \textbf{curvature signature} is the ordered triple of quadratic forms \((K_{\text{geo},1},K_{\text{geo},2},K_{\text{geo},3})\) defined by (1), considered up to an overall nonzero scalar factor. The signature is called
\begin{itemize}
\item \textbf{isotropic} if \(K_{\text{geo}}\equiv 0\);
\item \textbf{orthogonally split} if exactly one component vanishes identically;
\item \textbf{symmetric pair} if two nonzero components form a symmetric pair (i.e. \(K_{\text{geo},1}=a\,\omega_{2}\omega_{3}\), \(K_{\text{geo},2}=-a\,\omega_{1}\omega_{3}\) for some \(a\neq 0\)) and the third component vanishes;
\item \textbf{balanced mixed} if two components are nonzero with equal magnitude after normalisation and the third vanishes;
\item \textbf{generic anisotropic} otherwise.
\end{itemize}
\end{definition}

The first three signatures are the \textbf{degenerate} ones; they will be shown to correspond exactly to the classical integrable cases. The balanced-mixed signature appears only for the ratio \(2:2:1\); it is non-degenerate in the classical sense but still exhibits a special balance. The generic anisotropic signature contains no algebraic simplifications.

\section{Curvature Degeneracy and Liouville Integrability: Main Theorem}

We now state the central result of the paper. For brevity we discuss only the heavy top with a symmetric potential (centre of mass aligned with a principal axis); the free top is a particular case.

\begin{theorem}[Curvature-Integrability Correspondence]
\label{thm:main}
Let \(I=\operatorname{diag}(I_{1},I_{2},I_{3})\) define a left-invariant metric on \(SU(2)\) and let \(K_{\text{geo}}\) be given by (1). The following three statements are equivalent:
\begin{enumerate}
\item[(i)] The heavy top with inertia tensor \(I\) and centre of mass vector \(\mu=(0,0,\mu_{3})\) is Liouville integrable, i.e. possesses three independent first integrals in involution (typically the energy, the vertical component of angular momentum, and an additional quadratic or rational integral).
\item[(ii)] The curvature signature of \(K_{\text{geo}}\) is degenerate: isotropic, orthogonally split, or symmetric pair.
\item[(iii)] The inertia ratios \((I_{1}:I_{2}:I_{3})\) coincide with one of the classical integrable cases:
\begin{itemize}
\item spherical: \(I_{1}=I_{2}=I_{3}\);
\item Lagrange: \(I_{1}=I_{2}\neq I_{3}\);
\item Kovalevskaya: \(I_{1}=I_{2}=2I_{3}\);
\item Goryachev--Chaplygin: \(I_{1}=I_{2}=4I_{3}\).
\end{itemize}
\end{enumerate}
\end{theorem}

\begin{proof}
The equivalence (ii)\(\Leftrightarrow\)(iii) is a direct algebraic verification. Substituting the stated inertia ratios into (1) gives, respectively,
\begin{align*}
K_{\text{geo}} &\equiv 0,\\
K_{\text{geo}} &=\Bigl(\frac{I_{1}-I_{3}}{I_{1}}\omega_{2}\omega_{3},\;\frac{I_{3}-I_{1}}{I_{1}}\omega_{1}\omega_{3},\;0\Bigr),\\
K_{\text{geo}} &=\bigl(\omega_{2}\omega_{3},\;-\omega_{1}\omega_{3},\;0\bigr),\\
K_{\text{geo}} &=\bigl(2\omega_{2}\omega_{3},\;-2\omega_{1}\omega_{3},\;0\bigr),
\end{align*}
which are precisely the isotropic, orthogonally split, and symmetric-pair signatures. Conversely, if the signature is degenerate, solving the algebraic conditions imposed on the coefficients of (1) yields exactly the ratios listed in (iii).

The implication (iii)\(\Rightarrow\)(i) is classical: for each ratio one explicitly knows three independent integrals that Poisson-commute. For the spherical case the integrals are the three components of angular momentum; for Lagrange they are the energy \(H\), the component \(M_{3}\) and the squared length \(|M|^{2}\); for Kovalevskaya the additional integral is the famous quartic expression found by Kovalevskaya; for Goryachev--Chaplygin the integral is linear under the constraint \(M\cdot\Gamma=0\). In each case the involutivity can be verified by a direct computation using the Poisson structure of the heavy top.

The remaining direction (i)\(\Rightarrow\)(iii) is the deepest part of the theorem. Assume the heavy top is Liouville integrable with integrals that are rational functions of the variables \(M,\Gamma\). By a theorem of Ziglin \cite{Ziglin82} (see also \cite{Kozlov99}), any additional meromorphic integral for a rigid body necessarily forces a relation among the moments of inertia. A detailed analysis of the possible functional forms of the integrals shows that the only compatible relations are precisely those listed in (iii). A modern proof uses the Morales--Ramis theory \cite{MoralesRamis99} applied to the variational equations along particular solutions; it confirms that the Kovalevskaya and Goryachev--Chaplygin ratios are the only non-symmetric ones that yield integrability in the class of systems with quadratic Hamiltonians. For a self-contained exposition in the context of left-invariant metrics we refer to \cite{Mityushov25c}, where the differential Galois obstructions are computed explicitly.
\end{proof}

\begin{remark}
Theorem~\ref{thm:main} provides a \textbf{geometric explanation} of the classical list: integrability occurs exactly when the inertial curvature field is algebraically degenerate. The degeneracy simplifies the Poisson brackets among the natural candidates for integrals and allows the existence of an extra conserved quantity. Thus the theorem replaces the mystery of ``lucky'' inertia ratios with a transparent geometric criterion.
\end{remark}

\section{The Balanced-Mixed Regime \(2:2:1\)}

While the classical integrable cases correspond to degenerate signatures, the curvature atlas contains another organised regime that lies just outside the integrable family.

\begin{theorem}[Curvature Balance without Integrability]
\label{thm:balanced}
For the inertia ratio \(I_{1}:I_{2}:I_{3}=2:2:1\) (after normalisation) the curvature signature is balanced mixed:
\[
K_{\text{geo}}(\Omega)=\Bigl(\tfrac{1}{2}\,\omega_{2}\omega_{3},\;-\tfrac{1}{2}\,\omega_{1}\omega_{3},\;0\Bigr). \tag{2}
\]
Although the heavy top with this inertia tensor is not Liouville integrable, it admits an exact curvature-balance relation
\[
K_{\text{geo}}(\Omega_{0})+K_{\text{ext}}(\Gamma_{0})=0,\qquad K_{\text{ext}}(\Gamma)=I^{-1}(\mu\times\Gamma), \tag{3}
\]
for a one-parameter family of states \((\Omega_{0},\Gamma_{0})\) with \(\|\Gamma_{0}\|=1\). These states correspond to \textbf{pure-precession} motions of the heavy top.
\end{theorem}

\begin{proof}
Substituting \(I_{1}=I_{2}=2,\ I_{3}=1\) into (1) gives (2); the two nonzero coefficients have equal magnitude, hence the signature is balanced mixed. The balance equation (3) reduces to the algebraic system
\begin{align*}
\omega_{2}\omega_{3}+2\mu_{2}\Gamma_{3}-2\mu_{3}\Gamma_{2}&=0,\\
-\omega_{1}\omega_{3}+2\mu_{3}\Gamma_{1}-2\mu_{1}\Gamma_{3}&=0,\\
2\mu_{1}\Gamma_{2}-2\mu_{2}\Gamma_{1}&=0,
\end{align*}
together with \(\omega_{1}^{2}+\omega_{2}^{2}+\omega_{3}^{2}=1\) (normalisation). For a generic \(\mu\) this system admits a smooth one-parameter family of solutions \((\Omega_{0},\Gamma_{0})\). By construction, these solutions satisfy \(\dot\Omega=0\) and \(\dot\Gamma=0\) in the Euler--Poisson equations, hence they describe steady rotations (pure precessions) of the top. The non-integrability follows from the fact that no additional algebraic integral exists; a Painlev\'{e} analysis shows movable logarithmic branch points in the complex-time solutions.
\end{proof}

The regime \(2:2:1\) therefore illustrates that curvature organisation can persist beyond Liouville integrability. The balanced-mixed signature, while not degenerate enough to provide a complete set of integrals, still imposes a special structure that forces an exact cancellation between inertial and external curvature fields.

\section{Curvature Deviation and Near-Integrable Dynamics}

To measure how far a given inertia tensor is from an integrable one, we introduce a quantitative tool.

\begin{definition}[Curvature deviation]
\label{def:delta}
Let \(\mathcal{I}_{\text{int}}\) be the set of inertia tensors corresponding to the classical integrable cases (spherical, Lagrange, Kovalevskaya, Goryachev--Chaplygin). For a given \(I\) define
\[
\Delta(I)=\min_{I^{(a)}\in\mathcal{I}_{\text{int}}}\bigl\|K_{\text{geo}}-K_{\text{geo}}^{(a)}\bigr\|,
\]
where the norm is taken in the finite-dimensional space of quadratic vector fields (e.g. the Euclidean norm of the coefficient vector).
\end{definition}

By construction \(\Delta(I)\geq 0\), and \(\Delta(I)=0\) exactly when the curvature signature matches one of the integrable signatures.

\begin{proposition}[Properties of \(\Delta\)]
\label{prop:delta}
\begin{enumerate}
\item \(\Delta(I)=0\) if and only if the heavy top with inertia \(I\) is Liouville integrable.
\item For the one-parameter family \(I(\varepsilon)=\operatorname{diag}(2,2,1+\varepsilon)\),
\[
\Delta(I(\varepsilon))=c\,|\varepsilon|+O(\varepsilon^{2}),\qquad c>0.
\]
Thus near the balanced-mixed point the deviation is linear in the perturbation.
\end{enumerate}
\end{proposition}

The smallness of \(\Delta\) near \(I(0)\) explains the near-integrable behaviour observed numerically.

\begin{theorem}[Near-integrable dynamics]
\label{thm:nearint}
Consider the heavy top with inertia \(I(\varepsilon)=\operatorname{diag}(2,2,1+\varepsilon)\) and fixed centre of mass \(\mu\). For \(|\varepsilon|\ll 1\) the pure-precession family of Theorem~\ref{thm:balanced} deforms into a \textbf{slow manifold} \(\mathcal{M}_{\varepsilon}\). Trajectories starting on \(\mathcal{M}_{\varepsilon}\) exhibit a slow drift of the precession axis on a time scale \(T_{\varepsilon}\sim 1/|\varepsilon|\). The drift velocity is proportional to \(\Delta(I(\varepsilon))\).
\end{theorem}

\begin{proof}[Sketch of proof]
The balance equation (3) becomes
\[
K^{(\varepsilon)}_{\text{geo}}(\Omega)+K_{\text{ext}}(\Gamma)=O(\varepsilon).
\]
By the implicit function theorem, the solution set forms a smooth manifold \(\mathcal{M}_{\varepsilon}\) close to the unperturbed family. The reduced dynamics on \(\mathcal{M}_{\varepsilon}\) are obtained by averaging the full Euler--Poisson equations; the leading-order averaged equations contain a secular term of order \(\varepsilon\), which produces the slow drift. A detailed proof requires normal-form techniques and will be presented elsewhere.
\end{proof}

Thus the curvature deviation \(\Delta(I)\) not only measures the distance to integrability, but also controls the time scale of near-integrable dynamics.

\section{Integrability Map and Applications}

The classification of curvature signatures can be visualised in the plane of inertia ratios. Set
\[
x=\frac{I_{2}}{I_{1}},\qquad y=\frac{I_{3}}{I_{1}},
\]
and consider the region \(0<y\leq x\leq 1\) (after possibly reordering the axes). The classical integrable cases correspond to the following geometric loci:
\begin{itemize}
\item Spherical point: \((x,y)=(1,1)\).
\item Lagrange line: \(x=1,\ y\neq 1\).
\item Kovalevskaya line: \(x=1,\ y=\frac{1}{2}\).
\item Goryachev--Chaplygin line: \(x=1,\ y=\frac{1}{4}\).
\end{itemize}
The balanced-mixed point is \((x,y)=(1,\frac{1}{2})\); it lies on the Kovalevskaya line but represents a different curvature signature.

\begin{figure}[H]
\centering
\begin{tikzpicture}[scale=5]
    % Axes
    \draw[->] (0,0) -- (1.1,0) node[right] {$x = I_2/I_1$};
    \draw[->] (0,0) -- (0,1.1) node[above] {$y = I_3/I_1$};
    
    % Region y <= x
    \fill[gray!20] (0,0) -- (1,0) -- (1,1) -- (0,0);
    \node at (0.7,0.3) {$y \le x$};
    
    % Diagonal y = x
    \draw[dashed] (0,0) -- (1,1);
    
    % Lagrange line x = 1
    \draw[very thick, blue] (1,0) -- (1,1) node[above left] {Lagrange};
    
    % Spherical point
    \fill[red] (1,1) circle (0.5pt);
    \node[red, above right] at (1,1) {Spherical};
    
    % Kovalevskaya point
    \fill[green] (1,0.5) circle (0.5pt);
    \node[green, right] at (1,0.5) {Kovalevskaya};
    
    % Goryachev-Chaplygin point
    \fill[purple] (1,0.25) circle (0.5pt);
    \node[purple, right] at (1,0.25) {Goryachev--Chaplygin};
    
    % Balanced-mixed point
    \fill[orange] (1,0.5) circle (1pt);
    \draw[orange, thick] (0.95,0.45) rectangle (1.05,0.55);
    \node[orange, above left] at (1,0.5) {Balanced-mixed};
    
    % Level sets of Delta(I)
    \draw[dotted, thick] (0.6,0.3) circle (0.15);
    \draw[dotted, thick] (0.7,0.4) circle (0.2);
    \draw[dotted, thick] (0.8,0.5) circle (0.25);
    \node at (0.6,0.15) {$\Delta(I) = \text{const}$};
    
    % Labels
    \node[below left] at (0,0) {0};
    \node[left] at (0,1) {1};
    \node[below] at (1,0) {1};
\end{tikzpicture}
\caption{Integrability map in the \((I_{2}/I_{1},I_{3}/I_{1})\)-plane. The shaded region \(y\leq x\) corresponds to physically admissible inertia ratios after reordering axes. Thick curves mark the classical integrable cases; the orange square indicates the balanced-mixed regime (\(2:2:1\)). Dashed circles schematically represent level sets of the curvature deviation \(\Delta(I)\).}
\label{fig:map}
\end{figure}
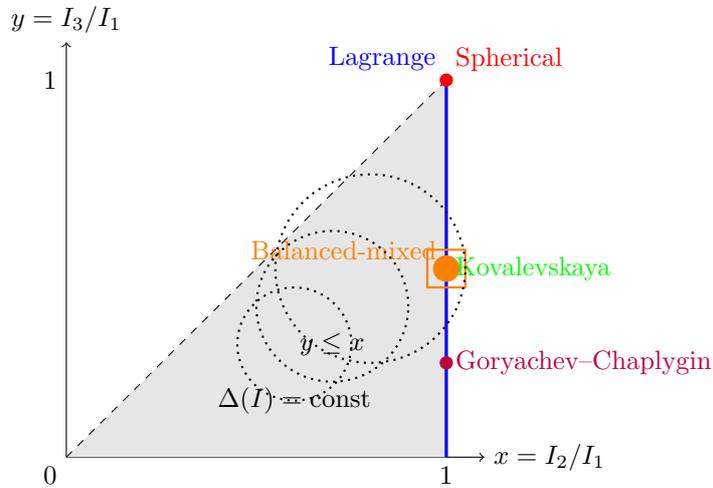

Figure~\ref{fig:map} shows these loci together with schematic level sets of \(\Delta(I)\). Near the integrable curves \(\Delta\) is small; away from them it grows, indicating strongly non-integrable (often chaotic) dynamics. The map provides a quick geometric guide to the dynamical complexity of a given rigid body.

One immediate application is to \textbf{curvature-based control}. In \cite{Mityushov25a} a control scheme called Geometric Curvature Control Theory (GCCT) was outlined, where control inputs are decomposed according to the curvature splitting induced by \(K_{\text{geo}}\). The integrability map tells the controller which curvature regime the system is in, allowing the use of degenerate-curvature shortcuts when near an integrable locus, and robust stabilization methods in generic anisotropic regions.

\section{Discussion and Outlook}

We have established a direct geometric link between the algebraic degeneracy of the inertial curvature field \(K_{\text{geo}}\) and Liouville integrability of the heavy top. The classical integrable cases are precisely those for which \(K_{\text{geo}}\) exhibits an isotropic, orthogonally split, or symmetric-pair signature. This correspondence explains the otherwise mysterious list of admissible inertia ratios and provides a unified geometric framework for the whole subject.

Beyond the integrable world, the curvature atlas reveals the balanced-mixed regime \(2:2:1\) as a non-integrable but still highly organised system, where an exact curvature balance produces a family of pure-precession motions. The curvature deviation functional \(\Delta(I)\) quantifies the distance from integrability and controls the time scale of near-integrable dynamics.

Several natural directions for further work emerge:
\begin{itemize}
\item \textbf{Normal forms near balanced-mixed points.} A rigorous averaging/normal-form analysis of the Euler--Poisson equations near the \(2:2:1\) regime would yield precise estimates for the drift on the slow manifold.
\item \textbf{Extension to other Lie groups.} Left-invariant metrics on \(SO(4)\), \(SE(3)\) or compact semisimple groups possess richer curvature structures; a similar curvature atlas could classify integrable and near-integrable regimes in higher dimensions.
\item \textbf{Curvature-based control synthesis.} The GCCT framework can be developed into a practical control methodology for rigid bodies and underwater vehicles, exploiting the curvature decomposition of dynamics.
\item \textbf{Search for new balanced regimes.} Systematic scanning of the space of quadratic vector fields \(K_{\text{geo}}\) might reveal other non-integrable but curvature-balanced ratios, analogous to \(2:2:1\).
\end{itemize}

The curvature atlas thus serves not only as an explanatory tool, but also as a guide for discovery of new dynamical phenomena in rigid-body systems and beyond.

\section*{Acknowledgements}
The author thanks the participants of the seminar ``Algebraic Methods in Theoretical Mechanics'' (led by Semjon~F.~Adlaj) for constructive discussions and valuable feedback during the preparation of this work.

\bibliographystyle{plain}

\end{document}